\documentclass[11pt]{article}
\usepackage[utf8]{inputenc}
\usepackage[margin=0.9657in]{geometry}
\usepackage[shortlabels]{enumitem}
\usepackage{bm}
\usepackage{bbm}
\usepackage{amsfonts}
\usepackage{amsmath}
\usepackage{xcolor}
\usepackage{graphicx}
\usepackage{amsthm}
\usepackage{qtree}
\usepackage{tree-dvips}
\usepackage{float}
\usepackage{authblk}
\usepackage{fontawesome5}
\usepackage{hyperref}

\newtheorem{theorem}{Theorem}
\newtheorem*{theorem*}{Theorem}
\newtheorem{lemma}{Lemma}
\newtheorem*{lemma*}{Lemma}
\newtheorem{definition}{Definition}
\newtheorem*{definition*}{Definition}
\newtheorem{corollary}{Corollary}
\newtheorem*{corollary*}{Corollary}

\newtheorem*{conjecture*}{Conjecture}

\newtheorem*{claim*}{Claim}

\newtheorem*{proposition*}{Proposition}

\newtheorem*{observation*}{Observation}

\newtheorem*{hypothesis*}{Hypothesis}

\usepackage[numbers]{natbib}

\newcommand{\B}{\{0,1\}}

\newcommand{\ber}{\mathrm{ber}}
\newcommand{\norm}[1]{\Vert #1 \Vert}
\newcommand{\harm}[1]{\bm{\breve}{#1}}
\def\cl{\bf}
\author{Oliver Korten \thanks{Institute for Advanced Study. Email: \href{mailto:korten@ias.edu}{korten@ias.edu}. This material is based upon work supported by the National Science Foundation under Grant No. DMS-2424441.}}
\date{}
\title{Top-Down Lower Bounds for All Depths}
\begin{document}
\maketitle

\begin{abstract}
We prove that Parity requires $2^{n^{\Omega(1)}}$ size De Morgan circuits of constant depth using a new method which is completely ``top-down'' in the sense of \cite{HJP1995}. The proof relies crucially on the core ideas developed in a line of work \cite{HJP1995,PPZ1999,MW2019,GRSS2024} which previously established top-down lower bounds for circuits of depth 3 and 4. We first present a proof of a lower bound $\exp(n^{3^{-d}})$. In this case, nearly all of the relevant combinatorial ideas necessary for the proof are already present in some form in \cite{GRSS2024}. 
We then present two extensions of this argument, the first achieving a lower bound $\exp(\epsilon_d n^{1/(2d-2)})$ for some $\epsilon_d>0$ depending only on $d$, and the second achieving the essentially tight lower bound $\exp(\epsilon_d n^{1/(d-1)})$. These improved results each hinge on establishing a key lemma which quantifies the extent to which a high entropy random variable in $\B^n$ will look close to uniform after projecting it onto a random small set of coordinates $R \subseteq [n]$.

\end{abstract}
\begin{quote}
\small
\textbf{AI Usage:} The result presented in Section~\ref{sec:original-proof}, giving the first top-down lower bounds for arbitrary depth $\cl{AC}^0$ circuits, was generated autonomously by GPT-6 Astra. This initial proof gave a lower bound of the form $\exp(n^{3^{-d}})$ for depth $d$ circuits. Subsequently, the author extended these methods to give a top-down proof of the stronger lower bound $\exp(n^{1/(2d-2)})$; this human-generated result is presented in Section~\ref{sec:new-lb}. The author then determined a combinatorial conjecture which would suffice to extend the arguments in Section~\ref{sec:new-lb} to the optimal $\exp(n^{1/(d-1)})$ bound, and asked the machine to prove it; the machine succeeded with some high-level human direction and the final tight lower bound is presented in Section~\ref{sec:nearly-tight}. A more detailed description of the AI methodology is given at the end of the introduction (\ref{subsec:roadmad-ai}).
All contents of this document were written solely by the human author, who takes full responsibility for their correctness.
\end{quote}

\section{Introduction}
A highly influential work of Karchmer and Wigderson \cite{KW1990} showed that lower bounds on the depth of circuits computing a function $f$ can be rephrased equivalently as lower bounds on the \emph{communication complexity} of a certain search problem associated with the function $f$ now called the ``KW game for $f$.'' A variant of this correspondence was shown to hold also for \emph{monotone computation}, and was immediately applied in \cite{KW1990} to obtain tight lower bounds on the depth of monotone circuits deciding undirected connectivity. Since \cite{KW1990}, a rich theory centered around \emph{lifting} has developed which supplies a very modular and general recipe for proving communication lower bounds for monotone KW games, and hence depth lower bounds for monotone circuits (the framework also yields strong lower bounds on monotone circuit size) \cite{RM1999,dRNV2016,PR2017,GP2018,GPW2018,GKRS2019,GGKS2020,dRMNPRV2020,LMMPZ2022,dRV2025}. In this framework, a lower bound is first proven for the monotone KW game in a simplified ``decision tree model,'' and then that simple lower bound is ``lifted'' to the more general communication-based model.

It has been a longstanding open problem to apply communication based arguments to obtain interesting lower bounds for \emph{non-monotone} circuits. In particular, it was pointed out in \cite{HJP1995} that at the time we did not even know how to reproduce \emph{known non-monotone lower bounds} in this framework. The authors of \cite{HJP1995} initiated the project of reproving the celebrated result of \cite{FSS1984,Ajtai1983} that $\mathrm{Parity} \notin \cl{AC^0}$ via communication-complexity arguments, which they refer to as ``top-down proofs'' (we will explain the meaning of ``top-down'' later on in the introduction). They proceeded to solve this problem in the particular case of depth 3 circuits. 
A few works have continued to study this question in the ensuing decades \cite{PPZ1999,MW2019,GRSS2024}; in particular, exponential lower bounds for depth 4 were obtained in \cite{GRSS2024}. The problem has remained open for any depth larger than 4. The current work completes this program, establishing exponential lower bounds for $\cl{AC^0}$ circuits of any fixed depth computing parity by a proof which is completely top-down in the sense intended in \cite{HJP1995}.

The value in proving an old result by a new method can be hard to judge on arrival. Our hope, and the hope of those who've previously studied this problem \cite{HJP1995,MW2019,GRSS2024}, is that such a proof will offer a bridge over which ideas from the highly-developed theory of monotone circuit complexity may start to be transferred over to the study of non-monotone circuit complexity, perhaps in the eventual service of making progress on the KRW conjecture \cite{KRW1995} and $\cl{P} \neq \cl{NC^1}$. Another motivation, highlighted in \cite{GRSS2024}, is the well-known problem of proving lower bounds for $\cl{PH^{cc}}$; this is a class extending $\cl{AC^0}$ for which random-restriction based arguments appear completely unavailable, and the authors in \cite{GRSS2024} argue that a top-down lower bound for $\cl{AC^0}$ may point us in the direction of such a lower bound.
Currently we have no clue as to whether the ideas in this work will play any role in future lower bounds of these kinds. At the very least, we believe the arguments here will offer some richer understanding of bounded depth circuits beyond what can be achieved by the method of random restrictions \cite{FSS1984,Ajtai1983,Yao1985,Hastad1986} or the method of polynomial approximation \cite{Razborov1987,Smolensky1987}.

\subsection{KW Games and Adversaries:}\label{subsec:adversary}
We start by defining Karchmer-Wigderson communication problems. For a pair of sets $X,Y \subseteq \B^n$, $X \cap Y = \emptyset$, the communication problem $\mathrm{KW}(X,Y)$ is defined as follows. Two parties Alice and Bob receive inputs $x \in X$ and $y \in Y$ respectively. Their goal is to determine some $i \in [n]$ such that $x_i \neq y_i$, by communicating with one another over the course of several rounds. We say that $\mathrm{KW}(X,Y)$ has a $d$-round protocol of cost $m$, if there is a deterministic protocol in which the players exchange messages for at most $d$ rounds, sending at most $m$ bits in every round, such that at the end of the protocol both parties agree on a valid output for $\mathrm{KW}(X,Y)$. In a given round, the next speaker is determined arbitrarily by the prior communication transcript. For a Boolean function $f$, we use the shorthand $\mathrm{KW}(f) := \mathrm{KW}(f^{-1}(0), f^{-1}(1))$.

\begin{lemma}[\cite{KW1990}]\label{lem:corresp}
A function $f: \B^n \to \B$ has a depth $\leq d$ unbounded fan-in De Morgan circuit in which each $\bigvee/\bigwedge$ node has $\leq 2^m$ children if and only if $\mathrm{KW}(f)$ has a $d$-round protocol of cost $m$.
\end{lemma}

As is well known, a communication protocol for a problem in which Alice receives inputs in $X$, and Bob in $Y$, corresponds to a recursive decomposition of $X \times Y$ into rectangles $X' \times Y'$, $X' \subseteq X, Y' \subseteq Y$ (see \cite{KN1997} for a comprehensive textbook on communication complexity). In a given round $i$ of a $d$-round protocol for $\mathrm{KW}(X,Y)$, the set of inputs consistent with the current communication transcript is a rectangle
$X' \times Y'$. If the next player to speak is Alice, her $m$-bit message performs an arbitrary partition of $X'$ into at most $2^m$ parts $X' = \bigcup_{j \leq 2^m} X'_j$. The remaining protocol can then be seen as giving, separately for each subrectangle $X'_j\times  Y'$, a $(d-i)$-round protocol of cost $m$ for $\mathrm{KW}(X'_j,Y')$. The existence of a zero-round protocol for some $X',Y'$ is equivalent to the existence of some $i \in [n]$ such that $x_i \neq y_i$ for all $x \in X', y \in Y'$.

With this perspective in mind, it is possible to prove that no $d$-round, cost $m$ protocol exists for $\mathrm{KW}(X,Y)$ by a kind of \emph{adversary argument}. Say that we define some combinatorial invariants $I_1,\ldots,I_{d+1}$ on rectangles in $X \times Y$, where invariant $I_{i}$ is meant to be satisfied prior to the $i^{th}$ round of communication, and $I_{d+1}$ to be satisfied after all communication is complete. If we can prove:
\begin{enumerate}
  \item (Base Case): The initial rectangle $X \times Y$ satisfies the initial round 1 invariant $I_1$.
  \item (Induction Step): If $X' \times Y'$ is any subrectangle satisfying the invariant $I_{i}$ for round $i$, and we partition one of the sides into at most $2^m$ parts e.g. $X' = \bigcup_{j \leq 2^m} X'_j$, then there exists some part $X'_j$ such that $X'_j \times Y'$ satisfies the invariant $I_{i+1}$ for round $i+1$.
  \item (Terminal Case): If $X' \times Y'$ satisfies the final invariant $I_{d+1}$, then for all $i \in [n]$ there exists $x \in X', y \in Y'$ such that $x_i = y_i$.
\end{enumerate}
then it follows that $\mathrm{KW}(X,Y)$ does not have a $d$-round protocol of cost $m$. We refer to this as an adversary argument, since we imagine the choice of part $X'_j$ in each inductive step as being selected by an adversary aiming to focus attention on the inputs in which the protocol fails to make sufficient progress.

\subsection{What Exactly Constitutes a ``Top-Down'' Proof?}
A ``top-down proof'' of a lower bound refers, informally speaking, to any communication lower bound for $\mathrm{KW}(f)$ which is proven by giving an \emph{explicit communication adversary} for $\mathrm{KW}(f)$. This is called top-down since tracing the evolution of a rectangle across the communication protocol corresponds precisely to tracing the set of inputs rejected/accepted by a gate in the circuit, starting from the output gate and following a downward path to some leaf computing a literal. This is in contrast to a ``bottom-up'' proof such as the random-restriction method, which eliminates gates of the circuit starting from the bottom and continuing upwards towards the output.

The substance of the term ``top-down proof'' rests entirely on what we mean by an ``explicit communication adversary.'' Indeed, the KW correspondence in Lemma~\ref{lem:corresp} is an exact equivalence, so \emph{any proof whatsoever} of $\mathrm{Parity} \notin \cl{AC^0}$ must in fact yield a lower bound on the communication complexity of $\mathrm{KW}(\mathrm{Parity})$, and hence implies the existence of \emph{some} adversary. 
However, if one takes any of the previously known proofs of $\mathrm{Parity} \notin \cl{AC^0}$ and works backwards through the correspondence, they will end up building an adversary whose next move is guaranteed only by a counterfactual, nonconstructive argument: assuming no next move for the adversary exists, we work backwards through the KW correspondence to deduce the existence of a circuit computing the function and contradict the known circuit lower bound. The reader may refer to \cite{KR2026} for an interesting analysis of such a ``non-explicit'' adversary arising from the approximation method \footnote{The work in \cite{KR2026} does not merely apply the KW correspondence generically to obtain an adversary in the way we describe, but proves that the approximation method yields an adversary with some particular features of interest. Nonetheless the adversary and its analysis remain non-constructive in roughly the same sense discussed here.}.
In contrast, our proof maintains a simple combinatorial invariant on the current rectangle, and directly supplies the adversary's next move based on this invariant.

\subsection{Notation:}
We introduce most notation as needed but start with some basic conventions that will be used throughout. All logarithms are base 2. For a finite coordinate set $I$ and $p\in [0,1]$, we use $P\subseteq_p I$ to denote a random $P \subseteq I$ which is sampled by including each $i \in I$ in $P$ independently with probability $p$. For a finite set $X$, we use $x \sim X$ to denote a uniformly random sample from $X$. For a random variable $\bm{X}$ with finite support, we use $H(\bm{X})$, $H_\infty(\bm{X})$ to refer to its Shannon entropy and min-entropy respectively. If $\bm{X}$ is supported on $\B^n$, we use $D_\infty(\bm{X}):= n - H_\infty(\bm{X})$ to refer to its \emph{min-entropy deficit}. For $X \subseteq \B^n$, we use $D_\infty(X):= n - \log |X|$; this is exactly the min-entropy deficit of the random variable $\bm{X}$ which is uniform on $X$, and hence we refer to this quantity also as the min-entropy deficit of the set $X$ itself.

\subsection{Roadmap and AI Methodology}\label{subsec:roadmad-ai}
In Sections \ref{sec:original-proof}, \ref{sec:new-lb}, and \ref{sec:nearly-tight}, we present top-down depth-$d$ circuit lower bounds of increasing quantitative strength: first $\exp(n^{3^{-d}})$, then $\exp(\epsilon_d n^{1/(2d-2)})$, and finally $\exp(\epsilon_d n^{1/(d-1)})$ which is tight up to the constant $\epsilon_d > 0$ (which in \cite{Hastad1986} can be replaced by an absolute constant $\epsilon > 0$). We describe here in a bit more detail the respective roles of the author and the machine assistant over the course of this work. 

Initially, the result in Section~\ref{sec:original-proof} was obtained completely autonomously by GPT-6 Astra. The author made some substantial effort to reorganize the proof in an intelligible way (the most substantive change being the introduction of the guiding function and the random variable $\bm{Y}^g$ to construct the mirror set, in place of a more inscrutable Markov argument in the original proof) but the underlying technical ingredients are essentially the same here as they appeared in the original machine-generated proof. In fact, the original proof found by Astra gave a lower bound of $\Omega_d(n^{1/F_d})$ where $F_d$ is the $d^{th}$ Fibonacci number ($F_d$ grows as $\Theta(\alpha^d)$ for some constant $\alpha \approx 1.6$); we chose to simplify the bound to $n^{3^{-d}}$ for the sake of a more straightforward exposition.

Subsequent to writing up the machine-generated results in Section~\ref{sec:original-proof}, the author found a way to extend these methods to obtain the improved lower bound in Section~\ref{sec:new-lb}; the introduction of the notion of $(p,k)$-limit, the reduction to the ``light patterns lemma,'' and the proof of the first form of the light patterns lemma via Fourier analysis were found by the author. GPT-6 Astra was used at various points to optimize parameters and simplify the presentation of the proof.

After the work in Section~\ref{sec:new-lb}, the author had determined that a sufficient strengthening of the ``light patterns lemma'' would be sufficient to obtain the near-optimal $\exp(\epsilon_d n^{1/(d-1)})$ lower bound using the same proof structure introduced by the author in Section~\ref{sec:new-lb}. With some high level steering by the author, the machine ultimately found a proof of this strengthened lemma. The author's contribution here was to suggest reusing the ``transference principle'' for downward closed sets (applied originally in Section~\ref{sec:original-proof} to prove the shattering lemma) and then to look for a more analytic analogue of the inductive proof of Pajor's lemma. Eventually the machine found a proof of the light patterns lemma which involved deriving a specialized reverse-hypercontractive inequality for decreasing functions on the $1/3$-biased hypercube. The author found this original proof to be conceptually opaque and worked for some time with the machine to simplify it, eventually arriving at the current argument in Section~\ref{sec:nearly-tight} involving the ``harmonic mean transform.'' 

\section{The Initial Lower Bound}\label{sec:original-proof}
In this section we present the initial proof generated by GPT-6 Astra of the lower bound $m \geq n^{3^{-d}}$ on the communication cost of $d$-round protocols for $\mathrm{KW}(X,Y)$ where $X,Y$ are the 0 and 1-inputs of the parity function on $n$ bits. 
\subsection{The Adversary}
\label{sec:game}
We start by describing the adversary and reduce its correctness to a single key combinatorial lemma (the ``mirror set lemma''). The following subsection is dedicated to a proof of the mirror set lemma.
 
\begin{theorem}
\label{thm:parity-game}
For any $n,d,m \in \mathbb{N}$, $\mathrm{KW}(\mathrm{Parity}_n)$ does not have a $d$-round communication protocol of cost $m$ when $m\leq  n^{3^{- d}}$, $n \geq 64^{3^d}$. Consequently any depth $d$ circuit computing $\mathrm{Parity}_n$ must have $> 2^{n^{3^{- d}}}$ wires whenever $n \geq 64^{3^d}$.
\end{theorem}

The proof will maintain a certain invariant on the rectangle $X\times Y$ which intuitively says it is hard to tell the two sets apart by looking at a random small set of coordinates. We make a note here that in the generic description of a communication adversary given in Subsection~\ref{subsec:adversary}, we may without loss of generality allow the adversary to restrict a selected rectangle $X_i\times Y$ to some nonempty subrectangle $X'_i\times Y$, where $X'_i\subseteq X_i$ (and symmetrically on the $Y$ side), and maintain some invariant defined by that subrectangle. We will use this convention as a matter of notational convenience. The crux of the proof is to show that, if the current rectangle satisfies the invariant and the next protocol message has small cost, then the adversary can select a subrectangle which satisfies a quantitatively weaker version of the invariant. In the following, for $x \in \B^n$, $P \subseteq [n]$ we use $[x]_P$ to denote the subcube $\{ z \in \B^n \mid z_i= x_i \; \forall i \notin P\}$ and $x_P \in \B^P$ to denote the projection of $x$ onto the coordinates in $P$.

\begin{definition}[$p$-limit Condition]
Say that $x\in \B^n$ is a $p$-limit of $Y \subseteq \B^n$ if 
\[
\Pr_{P \subseteq_p [n]} [ Y\cap [x]_P \neq \emptyset] \geq \frac{3}{4}
\]
For sets $X,Y\subseteq\B^n$, we say that the rectangle $X\times Y$ satisfies the $p$-limit condition on the left (resp. right) if every element $x\in X$ is a $p$-limit of $Y$ (resp. every $y \in Y$ is a $p$-limit of $X$).
\end{definition}

We note that the threshold $\frac{3}{4}$ is somewhat arbitrary; changing the threshold to any constant in $(0,1)$ would suffice for the main result. The notion of a ``$p$-limit'' is a minor variant of a definition used in \cite{GRSS2024} (see Section 2.2, equation (2) of \cite{GRSS2024}; their definition involves a uniform random $k$-set rather than an independent Bernoulli set). This in turn is a variant of the ``$k$-limit'' used earlier in \cite{HJP1995} which requires the stated condition to hold \emph{for all sets} $P$ of a given size; the authors of \cite{HJP1995} credit Sipser \cite{Sipser1984} with the introduction of the concept.

The invariant we will maintain on the rectangle $X\times Y$ is that (1) both $X$ and $Y$ are suitably large and (2) the rectangle satisfies either the left or right $p$-limit condition (for some largeness thresholds and $p$ values which decay as the rounds progress). As mentioned above, the key step is to show that after each protocol message, the adversary can select a subrectangle satisfying the next invariant. We make a few observations. First note that it is trivial to maintain condition (1), losing a factor $2^{-m}$ in largeness at each step, by choosing the largest component of the decomposition. Second, note that if the $p$-limit condition is satisfied on the left (resp. right) and the next message partitions $X$ (resp. $Y$) then the $p$-limit condition is trivially maintained, since it involves a universal quantifier over elements of $X$ (resp. $Y$). Hence, if it were possible to show that the left $p$-limit condition implies the right $q$-limit condition (possibly after passing to subsets of $X,Y$) and vice versa for some reasonable value of $q$, we would be done. This is exactly what the following ``mirror set lemma'' will accomplish.

\begin{lemma}[Mirror Set Lemma]
\label{lem:mirror}
Suppose $D_\infty(X)\leq k$, where $k\ge1$, $p>0$, and the rectangle $X\times Y$ satisfies the $p$-limit condition on the left. If $q = 64kp \leq 1$ then there exists a ``mirror set'' $Y'\subseteq Y$ such that the rectangle $X\times Y'$ satisfies the $q$-limit condition on the right, and $D_\infty(Y')\leq k+2pn+ 2$.
\end{lemma}

Note that the statement of this lemma immediately implies the symmetric form in which we swap the roles of $X,Y$, and so we use both forms freely.
We are now ready to prove Theorem~\ref{thm:parity-game}:
\begin{proof}[Proof of Theorem~\ref{thm:parity-game}]
Let $X_0 = \mathrm{Parity}^{-1}(0)$, $Y_0 = \mathrm{Parity}^{-1}(1)$. For $i \in\{ 0,\ldots, d\}$ let $k_i = n^{3^{i - d}}$,  $p_i = \frac{k_i}{4n}$. We observe the following:
\[
64 \leq k_0 \leq k_i \leq n \text{ for all } 0 \leq i \leq d, \quad m \leq k_0, \quad k_{i+1} = k_i^3 \text{ for all } 0 \leq i < d
\]
and from these it follows that, for all $0 \leq i < d$
\[
1.\; \;\; k_i + 2p_i n + 2 + m \leq 3 k_i + 2 \leq k_{i+1}, \quad\quad 2. \; \; \; 64 k_i p_i  \leq \frac{k_i^3}{4n} \leq p_{i+1} \leq \frac{1}{4}
\]

After $i$ rounds we maintain a rectangle $X_i\times Y_i$ such that the $p_i$-limit condition is satisfied (either on the left or right) and $D_{\infty}(X_i), D_{\infty}(Y_i) \leq k_i$. At the start, $k_0 = n^{3^{-d}} \geq 1 = D_{\infty}(X_0)= D_{\infty}(Y_0)$ so the deficit condition is satisfied. On the other hand it may be verified that for our parameters, a $p_0$-random set $P$ has $|P| > 0$ with probability $\geq \frac{3}{4}$; this immediately implies that the $p_0$-limit property is satisfied on both the left and right by the rectangle $X_0\times Y_0$ since for any string $z$ and any $|P|> 0$, $[z]_P$ will contain both even and odd parity strings. At the end we have $p_d  \leq \frac{1}{4}$ and $k_d \leq n$, which implies that the final rectangle $X_d\times Y_d$ does not admit a zero-round protocol: in this case the $p_d$-limit is satisfied on the left or right (say left) and both $X_d, Y_d$ are nonempty. For every $i$, a $p_d$-random set will contain $i$ with probability $\leq \frac{1}{4}$, and hence for any $x \in X_d$ and every $i \leq n$ we can find some $y \in Y_d$ agreeing with $x$ at index $i$.

It remains to show that if the rectangle $X_i\times Y_i$ satisfies the key invariant then we can reestablish this for $X_{i+1}\times Y_{i+1}$. Suppose the next message partitions $X_i = X_{i, 1} \cup \cdots \cup X_{i, 2^m}$ (the other case is handled symmetrically). If the $p_i$-limit condition holds for $X_i\times Y_i$ on the left then it trivially holds for every subrectangle $X_{i,j}\times Y_i$ (and hence so does the $p_{i+1}$-limit condition since $p_{i+1} \geq p_i$); we may take $j$ maximizing $|X_{i,j}|$ and set $X_{i+1} = X_{i,j}, Y_{i+1} = Y_i$. We have $D_{\infty}(X_{i+1}) \leq D_{\infty}(X_i) + m \leq k_{i+1}$ and $D_{\infty}(Y_{i+1}) = D_{\infty}(Y_i) \leq k_{i+1}$ so we are done. In the interesting case, the $p_i$-limit condition holds on the right, and we apply the mirror set lemma. We pass to some $X'_i \subseteq X_i$ so that $D_\infty(X'_i) \leq k_i + 2p_i n + 2$ and the rectangle $X'_i\times Y_i$ satisfies the $64 k_i p_i$-limit condition on the left (we are able to apply this lemma since $64 k_i p_i\leq \frac{1}{4}$ by (2)). We then consider the decomposition $X'_i = X'_{i,1} \cup \cdots \cup X'_{i,2^m}$ with $X'_{i,j} = X'_i \cap X_{i,j}$ and set $X_{i+1}$ to equal the largest component (and $Y_{i+1} = Y_i$). The deficit of $X_{i+1}$ is bounded by $k_i + 2p_i n + 2 + m $ and the rectangle $X_{i+1}\times Y_{i+1}$ satisfies the $64 k_i p_i$-limit condition. Since $k_{i+1} \geq k_i + 2p_i n + 2 +  m $ and $p_{i+1} \geq 64 k_i p_i$ by (1)/(2) we are done.
\end{proof}
It is instructive to compare the evolution of the random coordinate sets in this proof with that of the free coordinate sets in the iterative random restriction proof of H{\aa}stad \cite{Hastad1986}. The two arguments traverse the circuit in opposite directions. Here we need to choose random coordinate sets of expected size $\Omega(m)$, and allow their expected size to grow cubically at each stage, ending at $n/4$. Read backwards, the size of the random sets starts close to $n$ and shrinks cubically at each step. In the random-restriction proof, by contrast, each random restriction shrinks the number of free coordinates by a factor $\Theta(m)$. This more favorable parameter evolution yields the tight communication-cost lower bound $m = \Omega(n^{1/(d-1)})$; we will be able to (almost) match this parameter sequence later on in Sections~\ref{sec:new-lb} (growing the sets by a factor $\Theta_d(m^2)$) and \ref{sec:nearly-tight} (growing them by a factor $\Theta_d(m)$).

\subsection{Proof of the Mirror Set Lemma}
To prove the mirror set lemma we need two combinatorial lemmas which occur either directly, or in a slightly altered form, in \cite{GRSS2024}. The first appears roughly as Lemma 9 in \cite{GRSS2024}. We give a proof here which is noticeably simpler and yields a stronger bound. Lemma 9 in \cite{GRSS2024} is also sufficient to obtain the main result in this section, and indeed the original machine-generated proof used Lemma 9 of \cite{GRSS2024} as a black-box rather than our Lemma~\ref{lem:shattering} below. In the following, we say that a set $X\subseteq \B^n$ ``shatters'' $R \subseteq [n]$ if $\{ x_R \mid x \in X\} = \B^R$. We require a well-known result of Pajor which strengthens the classical Sauer-Shelah lemma \cite{Sauer1972,Shelah1972}:
\begin{lemma*}[Pajor's Lemma \cite{Pajor1985}]
Any $X\subseteq \B^n$ shatters at least $|X|$ distinct sets.
\end{lemma*}

\begin{lemma}[Shattering Lemma]
\label{lem:shattering}
If $X\neq \emptyset$ and $0\leq r\leq1/2$, then
\begin{equation*}
\Pr_{R\subseteq_r [n]}[X\text{ does not shatter }R]
\le 4rD_\infty(X).
\end{equation*}
\end{lemma}
\begin{proof}
Clearly we may assume $r> 0$. Below we reproduce a proof of a standard inequality \cite[Lemma~4.3.7]{Zhao2024}: for any downward-closed family $\mathcal{A} \subseteq 2^{[n]}$, $\Pr_{R \subseteq_r [n]}[R \in \mathcal{A}] \geq (\frac{|\mathcal{A}|}{2^n})^{4r}$. Apply this to the family $\mathcal{A}$ of sets shattered by $X$; by Pajor's Lemma \cite{Pajor1985}, we have $|\mathcal{A}| \geq |X|$, hence:
\[
\Pr_{R}[R \notin \mathcal{A}] \leq 1 - \Bigl(\frac{|\mathcal{A}|}{2^n}\Bigr)^{4r}\leq 1 - \Bigl(\frac{|X|}{2^n}\Bigr)^{4r} \leq 1 - 2^{- 4r D_{\infty}(X)} \leq 4r D_{\infty}(X)
\]
where we apply the bound $1 - 2^{-t} \leq t$ valid for all $t \geq 0$. It remains to prove the general inequality for downward-closed families. Let $h = \lfloor \frac{1}{2r} \rfloor$, so $h \geq \frac{1}{4r}$, and sample $h$ independent sets $R_1,\ldots,R_h \subseteq_r [n]$. Their union has the same distribution as a sample $A\subseteq_\theta[n]$, where $\theta:=1-(1-r)^h \leq hr \leq  \frac{1}{2}$. Note that the quantity $\Pr_{A \subseteq_\alpha [n]}[A \in \mathcal{A}]$ is decreasing in $\alpha$ since $\mathcal{A}$ is downward closed. We thus have:
\begin{gather*}
\frac{|\mathcal{A}|}{2^n}= \Pr_{A \subseteq_{\frac{1}{2}} [n]}[ A\in \mathcal{A}] \leq \Pr_{A \subseteq_\theta [n]}[ A\in \mathcal{A}] = \Pr[ \bigcup_i R_i \in \mathcal{A}] \leq \Pr[\bigwedge_i (R_i \in \mathcal{A})] = \Pr[R_1 \in \mathcal{A}]^h
\end{gather*}
\end{proof}

We also need the following well-known fact about min-entropy, whose proof can be found in \cite{GRSS2024} as a subclaim in the proof of Lemma 6. 
This ``entropy lemma'' says that if we condition on the value a random variable in $\B^n$ takes on the coordinates outside a set $S \subseteq [n]$, the entropy deficit on the remaining coordinates in $S$ will be at most the entropy deficit of the original random variable (on average).
\begin{lemma}[Entropy Lemma]
\label{lem:entropy}
For $S\subseteq[n]$ and $x \in X \subseteq \B^n$, let $X_{x,S}=X\cap[x]_S$, viewed inside its $|S|$-dimensional cube (in particular its min-entropy deficit is defined with respect to this subcube). Then
\begin{equation*}
\mathbb E_{x \sim X} D_\infty(X_{x,S})\le D_\infty(X).
\end{equation*}
\end{lemma}
\begin{proof}
Let $\bm{X}$ be uniform on $X$. Then
\begin{gather*}
\mathop{\mathbb{E}}\limits_{x \sim X}\! D_\infty(X_{x,S})  = |S| - H(\bm{X}_S \mid \bm{X}_{\overline{S}}) = |S| - H(\bm{X}) + H(\bm{X}_{\overline{S}}) \leq |S| - (n  - D_{\infty}(X)) + |\overline{S}| = D_{\infty}(X)
\end{gather*}
where the first equality uses the agreement of the entropies $H_{\infty}, H$ on uniform distributions, and the second equality uses the chain rule for Shannon entropy.
\end{proof}

At this point we are ready to prove the mirror set lemma, whose statement we reproduce for the pleasure of the reader:

\begin{lemma*}[Mirror Set Lemma (Lemma~\ref{lem:mirror}), Restated]
Suppose $D_\infty(X)\leq k$, where $k\ge1$, $p>0$, and the rectangle $X\times Y$ satisfies the $p$-limit condition on the left. If $q = 64kp \leq 1$ then there exists a ``mirror set'' $Y'\subseteq Y$ such that the rectangle $X\times Y'$ satisfies the $q$-limit condition on the right, and $D_\infty(Y')\leq k+2pn+ 2$.
\end{lemma*}

We describe the high-level structure of the proof before beginning formally. Say that $x,y \in \B^n$ are ``$P$-neighbors'' if  $y \in [x]_P$ (equivalently, $ x \in [y]_P$). When $x$ is a $p$-limit of $Y$, this means that the elements of $Y$ ``surround'' $x$ in the following geometric sense: with high probability over $P \subseteq_p [n]$, $x$ has some $P$-neighbor contained in $Y$. We aim to show that, assuming this holds and $X$ is suitably large, many points in $Y$ will also be surrounded by $X$ (where in this case surrounded $=$ $q$-limit). It is natural to try to prove this by a simple random process: sample $x \sim X$, $P \subseteq_p [n]$, and choose $y$ to be a random $P$-neighbor of $x$. If we could prove that $y$ was typically a $q$-limit of $X$, we would be part of the way there, but at this point we'd have no guarantee that this random process will ever generate some $y \in Y$. However, with an extra ``guiding step'' we can nudge this random step from $X$ so that it frequently falls into $Y$. Define a \emph{guiding function} $g: X \times 2^{[n]} \to \B^n$ to be any function such that $g(x,P)\in [x]_P$ for all $x,P$. Consider choosing a guiding function $g$ which selects $g(x,P)$ to lie in $Y$ whenever $[x]_P \cap Y$ is nonempty, and which is set to some arbitrary element of $[x]_P$ otherwise. We now consider the \emph{guided random step}: sample $x,P$ and walk to $g(x,P) \in [x]_P$. Since every element of $X$ is a $p$-limit of $Y$, this random step will land in $Y$ with probability $\geq \frac{3}{4}$. If we can show that this guided random output $y$ is usually a $q$-limit of $X$, then we will have established the existence of some $y\in Y$ which is a $q$-limit of $X$. Finally, if we can prove that this guided random step has sufficiently high min-entropy, we can conclude that \emph{many} elements of $Y$ are $q$-limits of $X$ and the lemma will be proven.

At this point, we have reduced the lemma to proving a statement \emph{purely about} $X$. Namely, assuming $X$ is large, if we take any guiding function $g$ and sample $x \sim X, P \subseteq_p [n]$, with high probability we have that the guided random step $g(x,P)$ is a $q$-limit of $X$. It suffices to show that with high probability over $x,P$, \emph{every element of} $[x]_P$ is a $q$-limit of $X$; this will handle all possible guiding functions. Moreover we need to establish that any such guided random step has suitably high entropy (so that we may obtain many different $q$-limits of $X$ in $Y$ in the end). The latter entropic claim follows from a simple double-counting argument. The crux of the entire proof is showing the former claim, that every $P$-neighbor of $x$ is a $q$-limit of $X$ with high probability over $x\sim X,P\subseteq_p [n]$. 

 \begin{proof}[Proof of Mirror Set Lemma]
We follow the outline above. Fix any guiding function $g(x,P) \in [x]_P$ and consider the random variable $\bm{Y}^g$ defined as follows: sample $x \sim X, P \subseteq_p [n]$, and output $y = g(x,P)$. We establish the following two key claims for all $g$:

\[
1. \; \; \Pr[\bm{Y}^g \text{ is a }q\text{-limit of }X] \geq \frac{1}{2} \quad\quad\quad\quad 2. \; \; H_{\infty}(\bm{Y}^g) \geq n - k - 2pn
\]
Given these the lemma follows directly. Set $g(x,P)$ to be any element of $Y \cap [x]_P$ if $Y \cap [x]_P \neq \emptyset$, otherwise define it to be some arbitrary element of $[x]_P$. By the $p$-limit condition for the rectangle $X\times Y$ on the left, we have that
$\Pr[\bm{Y}^g \in Y] \geq \frac{3}{4}$. Combining this with (1), we have that $\Pr[\bm{Y}^g \in Y \land \bm{Y}^g \text{ is a }q\text{-limit of }X] \geq \frac{1}{4}$. Applying (2), $H_\infty(\bm{Y}^g \mid \bm{Y}^g \in Y \land \bm{Y}^g \text{ is a }q\text{-limit of }X ) \geq n - k - 2pn - 2$. If we now take $Y'$ to be those elements $y$ in the support of $\bm{Y}^g$ satisfying ``$y \in Y \land y \text{ is a }q\text{-limit of }X$'', then we must have $D_{\infty}(Y') \leq k + 2pn + 2$ and we are done. 

We establish (2) first since it's simpler. Let $y \in \B^n$. The key point is that, conditioned on having sampled a particular set $P$ during the generation of $\bm{Y}^g$, there are at most $2^{|P|}$ possible values of $x$ that could have produced a given output $y$.  Hence
\begin{gather*}
\Pr[\bm{Y}^g = y] = \mathbb{E}_P \Pr[\bm{Y}^g = y \mid P \text{ is sampled}] \leq  \frac{1}{|X|} \mathbb{E}_P 2^{|P|} 
= \frac{1}{|X|} \prod_{i \leq n} (\mathbb{E}_{b \sim \ber(p)} 2^{b}) = \frac{(1+p)^n}{|X|}
\end{gather*}
so $H_{\infty}(\bm{Y}^g) \geq \log |X| - n \log (1+p) \geq n - k - 2pn$.

We now establish (1). By Markov, it suffices to show 
\[
\Pr_{x,P,Q}[\;[g(x,P)]_Q \cap X \neq \emptyset] \geq \frac{15}{16}, \text{ which follows from } \Pr_{x,P,Q}[\forall  y \in [x]_P, [y]_Q \cap X \neq \emptyset] \geq \frac{15}{16}
\]
where $x \sim X, P \subseteq_p [n], Q \subseteq_q [n]$ are sampled independently in both of the above expressions. We interpret the inequality on the right above. Let $\bm{X}$ be the random variable distributed uniformly on $X$. We sample two random sets $P\subseteq_p [n], Q \subseteq_q [n]$, $|Q| \gg |P|$ (typically), and want to show that if we reveal the coordinates of $\bm{X}$ \emph{outside} $P \cup Q$, then for every possible modification of the coordinates in $P$, there is some subsequent modification to the coordinates in $Q$ which leads us back into the support of $\bm{X}$. Since the subsequent modifications to $Q$ allow us to overwrite any initial modification to $P \cap Q$, we may assume that the first step only modifies $P \setminus Q$. Then in terms of the random variable $\bm{X}$, this is saying that if we learn the bits outside of $P \cup Q$, with high probability we will not learn \emph{anything definitive} about the possible values of the random variable on the coordinates in $P\setminus Q$: every candidate value in $\B^{P \setminus Q}$ will still have positive probability. We can phrase this directly in the language of shattering: if there exists $y \in [x]_P$ such that 
$[y]_Q \cap X = \emptyset$, this means that the set $X_{x, P \cup Q}$ does not shatter $P \setminus Q$, where $X_{x, P \cup Q} = X \cap [x]_{P \cup Q}$ is defined as in the statement of the entropy lemma. So our goal is precisely to prove: 
\[
\Pr_{x,P,Q}[X_{x, P \cup Q} \text{ does not shatter }P \setminus Q] \leq \frac{1}{16}\]
Now, choose some $x\in X$ and $S \subseteq [n]$ and consider the above probability when we fix this value of $x$ and condition on $P \cup Q = S$. Note that, conditioned on $P \cup Q = S$, $P \setminus Q$ has the same distribution as a sample $R\subseteq_r S$ where $r := \frac{ p(1-q)}{p + q - pq} \leq \frac{p}{q} \leq \frac{1}{64k}$ (remember $q = 64 k p$ in the statement of the lemma). Now, we apply the shattering lemma to $X_{x,S}$ (considered to lie inside $\B^S$). Since $k \geq 1$ we have $r < \frac{1}{2}$ so we may soundly apply it and conclude
\[
\Pr_{P, Q}[ X_{x, P \cup Q} \text{ does not shatter } P \setminus Q | P \cup Q = S] = \Pr_{R \subseteq_r S}[ X_{x,S} \text{ does not shatter } R] \leq 4 r D_{\infty} (X_{x,S})
\]
We then use the entropy lemma to reason that $D_{\infty}(X_{x,S})$ is bounded by $D_\infty(X)$ on average. We have:
\begin{gather*}
\mathbb{E}_{S} \Pr_{x, P, Q}[ X_{x, P \cup Q} \text{ does not shatter } P \setminus Q \mid P \cup Q = S] \\\leq 4 r\mathbb{E}_{x,S} D_{\infty}(X_{x,S}) \leq 4 r D_{\infty}(X) \leq 4k r \leq \frac{4k}{64k} \leq \frac{1}{16}
\end{gather*}

\end{proof}

\section{A Stronger Lower Bound}\label{sec:new-lb}

The primary source of the suboptimal parameters in the lower bound in the previous section is the additive $pn$ penalty in the deficit of the mirror set $Y'$ in the mirror set lemma. This required us to take each $p_i n$ to be at least $p_{i-2}n \cdot p_{i-1} n$ and led to a lower bound $m \geq n^{\exp(-\Theta(d))}$ on communication cost. Indeed, if it were possible to prove an extension of the mirror set lemma, where all aspects of the lemma statement are the same, but the lower bound on the size of the mirror set is improved to $D_{\infty}(Y') \leq O(k)$, this would be sufficient to obtain an $m \geq \Omega_d ( n^{1/d})$ lower bound on communication. In this case, during the inductive depth $d$ argument, we could take $k_{i+1} \leq O(k_i) + m$ to be the current entropy loss at round $i+1$, and hence maintain $k_i = O_d(m)$ throughout. Since each application of our hypothetical extended mirror set lemma would merely require $p_{i+1} = O(k_i p_i)$, starting with $p_0 = \frac{2}{n}$ we would have $p_d \leq O_d( \frac{m^d}{n})$, and hence $p_d \leq \frac{1}{4}$ whenever $m \leq \epsilon_d n^{1/d}$ for a suitable $\epsilon_d >0$. With slightly more care (observing that the limit condition is satisfied on both sides at the start) this argument could be improved to $m \geq \Omega_d(n^{1/(d-1)})$, which would be optimal up to the hidden constant in $\Omega_d(\cdot)$ depending on $d$.

Unfortunately this proposed extension of the mirror set lemma is false for a quite basic reason: under the assumptions of the lemma (namely, $D_{\infty}(X) \leq k$ and everything in $X$ is a $p$-limit of $Y$), it is not possible to bound $D_\infty(Y)$ by any function of $k$. We thus can't even get the desired lower bound on the density of $Y$ itself, let alone the density of the elements in $Y$ which are $q$-limits of $X$. An example is as follows: let $[n]$ be partitioned into $\sqrt{n}$ blocks of length $\sqrt{n}$, let $Y$ consist of the strings whose parity in every block is even, and let $X$ be its complement. Then $D_\infty(X) \leq 1$, and for $p = \frac{\log n}{\sqrt{n}}$, the $p$-limit condition for $X\times Y$ on the left is satisfied: a random set $P \subseteq_p [n]$ will intersect every block with probability $\geq 1 - \sqrt{n}(1- \frac{\log n}{\sqrt{n}})^{\sqrt{n}} \geq \frac{3}{4}$. On the other hand, $D_{\infty}(Y) = \sqrt{n} \gg 1$. 

In the above counterexample, we had a set $Y$ which was both (1) sparse and (2) had many $p$-limits. Upon closer inspection, for these $p$-limits $x$, we can observe that for most $P \subseteq_p [n]$, only a very small fraction of $y \in [x]_P$ lie in $Y$. In order to get around this kind of example, it might therefore suffice to change the requirement on a limit point $x$ and demand that for a typical $P$, $Y \cap [x]_P$ contains a decently large fraction of the subcube $[x]_P$. This motivates the following definition:
\begin{definition}
For $x \in \B^n, Y \subseteq \B^n$, we say that $x$ is a $(p,k)$-limit of $Y$ if 
\[
\Pr_{P \subseteq_p [n]} [|[x]_P \cap Y| \geq 2^{|P|-k}] \geq \frac{3}{4}
\]
\end{definition}
The left and right $(p,k)$-limit conditions for a rectangle $X\times Y$ are defined exactly as before. 
The reader may guess by our choice of variable names that we will consider the $(p,k)$-limit condition for sets with min-entropy deficit $\approx k$. As a sanity check, we can at least confirm that the previous kind of counterexample no longer exists: if $D_\infty(X) \leq k$, and $X\times Y$ satisfies the $(p,k)$-limit condition on the left, then $D_{\infty}(Y) \leq O(k)$. In particular:
\[
\frac{|Y|}{2^n} =\mathop{\mathbb{E}}_{\substack{z \sim \B^n\\ P \subseteq_p [n]}} \frac{|Y \cap [z]_P|}{2^{|P|}} \geq \frac{3}{4}\frac{|X|}{2^n}2^{-k} \geq \frac{3}{4}2^{-2k}
\]
So at the very least, we have ruled out the possibility that a set $Y$ can have deficit much larger than $k$ while having at least $2^{n-k}$ many $(p,k)$-limits. It remains to show that in fact, we can obtain some $Y' \subseteq Y$ of comparable size consisting entirely of $(O(kp), O(k))$-limits of $X$. In this section we will not be able to achieve this exactly, but instead obtain a subset of $Y$ with deficit $O(k)$, consisting of $(O(k^2p), O(k))$-limits of $X$; this will lead to a bound of the form $\exp(n^{1/(2d-2)})$ rather than $\exp(n^{1/(d-1)})$.

\begin{lemma}[Strong Mirror Set Lemma]\label{lem:strong-mirror}
Let $p > 0, k \geq 1$. Say that $X,Y \subseteq \B^n$, $X\times Y$ satisfies the $(p,k)$-limit condition on the left, and $D_{\infty}(X) \leq k$. Assume $q = Ck^2p \leq \frac{1}{2}$, where $C$ is some universal constant. Then there exists $Y' \subseteq Y$, $D_{\infty}(Y') \leq 2k + 2$, such that $X\times Y'$ satisfies the $(q,65k)$-limit condition on the right.
\end{lemma}

The high level structure of the argument will be the same; the main difference in the proof is that we must replace the shattering lemma (Lemma~\ref{lem:shattering}) with the following:
\begin{lemma}[Light Patterns]\label{lem:light-patterns}
Let $ k \geq 1, r \in [0,1]$, and let $\bm{X}$ be a random variable in $\B^n$, $D_\infty(\bm{X}) \leq k$. There is a universal constant $C$ such that the following holds whenever $r \leq (Ck)^{-2}$: with probability $\geq \frac{63}{64}$ over $R \subseteq_r [n]$ we have:
\[
|\{ z \in \B^R \mid \Pr[\bm{X}_R = z] \leq 2^{- |R| - 1}\}| \leq 2^{|R|- k }
\]
\end{lemma}
We will defer the proof of this lemma to the following subsection, and proceed now to the proof of the strong mirror set lemma under the assumption that it is true. 

\begin{proof}[Proof of the Strong Mirror Set Lemma]
Define a ``$k$-guiding function'' to be a function which selects, for each $x\in X$, $P \subseteq [n]$, a set $G(x,P) \subseteq [x]_P$ of size $|G(x,P)| = \lceil 2^{|P| - k}\rceil $ (compare to the guiding function from the proof of the original mirror set lemma, which selected a single element of $[x]_P$). We now consider the random variable $\bm{Y}^G$ defined as follows: sample a uniform $x \sim X$, $P \subseteq_p [n]$, a uniform $y \sim G(x,P)$, and output $y$. By the left $(p,k)$-limit condition on $X\times Y$, we know that for an appropriate choice of $G$, we have $\Pr[\bm{Y}^G \in Y] \geq \frac{3}{4}$. We want to show that (1) $H_{\infty}(\bm{Y}^G) \geq n - 2k $ and (2) with probability $\geq \frac{1}{2}$, $\bm{Y}^G$ is a $(q, 65k)$-limit of $X$. If we have both of these, we may take $Y'$ to be the set of $(q,65k)$-limits of $X$ lying in $Y$; combining (1)/(2) and the fact that $\Pr[\bm{Y}^G \in Y] \geq \frac{3}{4}$ guarantees that $D_{\infty}(Y') \leq 2k + 2$ and we will be done.

The proof of (1) was already sketched in the beginning of this section: for each $y \in \B^n$ we have:
\begin{gather*}
\Pr[\bm{Y}^G = y] = \mathbb{E}_P \frac{1}{|X|}\sum_{x \in X} \frac{\bm{1}\{ y \in G(x,P)\} }{|G(x,P)|} \leq \mathbb{E}_P \frac{2^{|P|}}{|X|} 2^{k - |P|} \leq \frac{2^k}{|X|} \leq 2^{2k-n}
\end{gather*}
Again following the proof of the mirror set lemma we apply a Markov argument and reduce (2) to showing
\setcounter{equation}{2}
\begin{gather}
\Pr_{x,P,Q, y \sim G(x,P)} \Bigl[ |[y]_Q \cap X| \geq 2^{|Q| - 65k} \Bigr] \geq \frac{15}{16} \label{eq:1}
\end{gather}
where $x \sim X, P \subseteq_p [n], Q \subseteq_q [n]$ are sampled independently. We claim that (\ref{eq:1}) in turn follows from showing that, for the set $A(x,P,Q):= \{y \in [x]_P \mid |[y]_Q \cap X| < 2^{|Q|-65k} \}$, we have
\begin{gather}
\Pr_{x,P,Q} \Bigl[ |A(x,P,Q)| \leq 2^{|P| - k - 6} \Bigr] \geq \frac{31}{32} \label{eq:2}
\end{gather}
Indeed, if (\ref{eq:2}) holds, then with probability at least $\frac{31}{32}$ over $x,P,Q$, we will have that a random $y \sim G(x,P)$ lands in $A(x,P,Q)$ with probability at most $2^{-6}$, and hence the overall probability over $x,P,Q,y \sim G(x,P)$ that $|[y]_Q \cap X| < 2^{|Q|-65k}$ is bounded by $\frac{1}{32} + 2^{-6} \leq \frac{1}{16}$. 

We interpret the quantity in (\ref{eq:2}). We want to show that, with probability $\geq \frac{31}{32}$ over the choice of $x,P,Q$ the following holds: for at least a $(1 - 2^{-k-6})$ fraction of the possible modifications we can make to $x$ on the coordinates in $P$, at least a $2^{-65k}$ fraction of \emph{subsequent} modifications to the variables in $Q$ will cause us to end up in $X$. To simplify the picture a bit, we observe that in the first phase we may instead count the fraction of modifications made to variables in $P\setminus Q$, since we may freely reassign the variables in $P \cap Q$ in the second stage, and every assignment in $P \setminus Q$ has exactly $2^{|P \cap Q|}$ extensions along the variables in $P$. 
Let $R = P \setminus Q$ and $S = P \cup Q$ in what follows. Say that in the first step we modified the coordinates in $R$ to some pattern $w \in \B^{R}$. We then want to lower bound the fraction of patterns $z \in \B^Q$ such that the combined string using $z$ on the bits in $Q$, $w$ on the bits in $R$, and $x_{\,\overline{S}}$ on the bits outside $S$ lies in $X$; in particular we want to lower bound this fraction by $2^{-65k}$. For the set $X_{x,S} = X \cap [x]_{S}$, we can observe that this ``fraction of $Q$-patterns leading us to $X$'' is given by:
\begin{gather}
\frac{|\{ x' \in X_{x,S} \mid x'_R = w\}|}{2^{|Q|}} = \frac{|X_{x,S}|}{2^{|Q|}}\Pr[(\bm{X}_{x,S})_{R} = w]\\
=   2^{|S|- D_\infty(X_{x,S}) - |Q|} \Pr[(\bm{X}_{x,S})_{R} = w]=  2^{|R|- D_\infty(X_{x,S})} \Pr[(\bm{X}_{x,S})_{R} = w]
\end{gather}
where $\bm{X}_{x,S}$ is uniform on $X_{x,S}$.

At this point, we proceed as in the mirror set lemma, considering the quantity in (\ref{eq:2}) when we fix a particular value of $x$ and $S$ and condition on $P \cup Q = S$. Say that $(x,S)$ are ``good'' if $D_{\infty}(\bm{X}_{x,S}) \leq 64k$. Under this conditioning, $R = P \setminus Q$ is sampled as $R\subseteq_r S$ for some $r \leq \frac{p}{q}$ (this is the same calculation from the original mirror set lemma). 
Under the assumption that $(x,S)$ are good we may apply the ``light patterns lemma'' (Lemma~\ref{lem:light-patterns}) with min-entropy deficit parameter $64k$ to the random variable $\bm{X}_{x,S}$ in $\B^S$, and conclude: with probability $\geq \frac{63}{64}$ over $R$, for all but a $2^{-64k} \leq 2^{-k-6} $ fraction of the possible assignments $w \in \B^R$, $\Pr[(\bm{X}_{x,S})_R = w] \geq 2^{-|R|-1}$. 
We may safely apply this lemma provided $r \leq (C64k)^{-2}$ for a suitably large constant $C$, which the assumptions of the lemma guarantee. Now, overall we get that whenever $(x,S)$ are good, with probability $\geq \frac{63}{64}$ over $R$, for all but a $2^{-k-6}$ fraction of patterns $w \in \B^{R}$, the fraction of $Q$ patterns leading us into $X$ is at least
\[
2^{|R|-D_{\infty}(X_{x,S})} 2^{-|R|-1} \geq 2^{ - 64k - 1}\geq 2^{-65k}
\]
We conclude with an application of the entropy lemma (Lemma~\ref{lem:entropy}), which (using an additional Markov argument) tells us that $(x,S)$ are good with probability at least $\frac{63}{64}$. Overall we determine that the probability in (\ref{eq:2}) is lower bounded by $\frac{63}{64}$ for any fixing of $(x,S)$ which are good, and the probability that $(x,S)$ are good is at least $\frac{63}{64}$, from which (\ref{eq:2}) follows by a union bound.

\end{proof}
From this we immediately obtain our improved communication adversary/circuit lower bound:
\begin{theorem}
\label{thm:strong-parity-game}
For each $d\geq 2$ there is $\epsilon_d>0$ such that
$\mathrm{KW}(\mathrm{Parity}_n)$ does not have a $d$-round communication protocol of cost $m$ when $m\leq \epsilon_d n^{1/(2d-2)}$, provided $n$ is sufficiently large. Consequently any depth $d$ circuit computing $\mathrm{Parity}_n$ must have $>2^{\epsilon_d n^{1/(2d-2)}}$ wires.
\end{theorem}
\begin{proof}[Proof Sketch]
The proof is essentially identical to that of Theorem~\ref{thm:parity-game}, using our new limit condition and the strong mirror set lemma in place of the old condition and old lemma. We indicate merely how the parameters must be set. Let $p_0 =p_1= \frac{2}{n}$, $k_i = 65^i m$, and for $i > 0$, $p_{i+1} = Ck_i^2 p_i$ for a suitable constant $C$. We maintain that, if $X_i \times Y_i$ is the current rectangle after round $i$, then it satisfies the $(p,k)$-limit condition (either on the left or right) for some $p \leq p_i, k \leq k_i$, and that $D_\infty(X_i), D_\infty(Y_i) \leq k_i$. Initially, $X_0 \times Y_0$ satisfies the limit condition on both sides, hence we can get through the first round without applying the mirror set lemma which allows us to take $p_0 = p_1$.
\end{proof}

\subsection{Random Projections of a Dense Set Are Close to Uniform}

The shattering lemma from Section~\ref{sec:original-proof} tells us the following: if $\bm{X}$ is a high min-entropy variable in $\B^n$, and $R$ is a random small set of coordinates, then the random variable $\bm{X}_R$ (the projection of $\bm{X}$ onto $R$) will be close to the uniform distribution with high probability over $R$, in the following weak sense: the support of $\bm{X}_R$ equals that of the uniform distribution (namely, all of $\B^R$). In particular, the quantitative statement says that if $R \subseteq_r [n]$, then $\bm{X}_R$ ``looks uniform'' in this sense with probability $\geq 1 - O(r D_{\infty}(\bm{X}))$. Note that this is essentially tight and matches what will occur when $\bm{X}$ is uniform on a subcube with, say, the first $k$ coordinates fixed to 0 (so $D_\infty(\bm{X})=k$). In this case we will have that $\bm{X}_R$ is exactly uniform whenever $R \cap \{ 1,\ldots, k \} = \emptyset$, which happens with probability $\geq 1- O(kr)$. Another more basic lemma of the same flavor is Shearer's lemma (which was in fact used in the proof of the weaker shattering lemma in \cite{GRSS2024}), which once again says that under the same conditions, $\bm{X}_R$ will be close to uniform for most $R$; in this case, ``close to uniform'' is quantified by Shannon entropy. In particular, we can conclude that with probability $\geq 1 - O(r D_{\infty}(\bm{X})) $ over $R$, $H(\bm{X}_R) \geq |R| - 0.1$, in other words we've only learned a fraction of a bit of Shannon information about $\bm{X}_R$. 

In this section we present two more results of the same form, using different notions of closeness to uniformity. 
Our goal is to prove the ``light patterns'' lemma (Lemma~\ref{lem:light-patterns}), which says that aside from a very small number of exceptions, every pattern in $\B^R$ will occur in $\bm{X}_R$ with at least half its probability mass under the uniform distribution. While the ideal version of this result would require only that $r D_{\infty}(\bm{X}) \ll 1$ (matching the case of subcubes, and the shattering/Shearer lemmas), in this section we are only able to establish it in the regime $r D_\infty(\bm{X})^2 \ll 1$. We will prove Lemma~\ref{lem:light-patterns} as a simple corollary of a more general result that we call the \emph{uniformity lemma} which says that 
with high probability over $R$, the $\norm{\cdot}_t$ distance of the (appropriately normalized) probability mass function of $\bm{X}_R$ from uniform is small. Our proof is inspired by the Fourier-analytic proof of the \emph{uniform marginals lemma} in \cite{GPW2020} and its subsequent simplification using the level-$\ell$ inequality in \cite{Wu2018}. A key step in the proof of the uniform marginals lemma is to show that, for high entropy $\bm{X}$ and a random set $R$, $\bm{X}_R$ is close to uniform in $L_\infty$ distance \emph{when the distribution is averaged over $R$}. This is incomparable to our task; on the one hand we require a stronger statement about the distribution of $\bm{X}_R$ for most \emph{fixed values} of $R$, but on the other hand our notion of distributional closeness is weaker. More precisely, in the terminology introduced below, \cite{GPW2020,Wu2018} aim to bound $\norm{\mathbb{E}_R f_R - 1}_\infty$  \footnote{In their setting, the random set $R$ has a fixed size $m$ and there is a canonical identification of each projection $x \mapsto x_R$ with a map $\B^n \to \B^m$; this makes $\mathbb{E}_R f_R$ well-defined.}, while our goal is to bound $\norm{f_R - 1}_t$ for most $R$, with $t < \infty$.

We begin with some standard definitions from the analysis of Boolean functions; see \cite{ODonnell2014} for a comprehensive textbook on the subject.
\begin{definition}
For $S \subseteq [m]$, the $S^{th}$ Fourier character $\chi_S: \B^m \to \{ -1,1\}$ is given by $\chi_S(z) = (-1)^{\sum_{i \in S} z_i}$. For a function $f: \B^m \to \mathbb{R}$, the $S^{th}$ Fourier coefficient of $f$ is given by $\widehat{f}(S) = \mathbb{E}_{z \sim \B^m} f(z) \chi_S(z)$. We use $\norm{f}_t:= \bigl( \mathbb{E}_{z \sim \B^m} |f(z)|^t \bigr)^{1/t}$ to denote the (normalized) $L_t$ norm of $f$.
\end{definition}

As in the case of Shearer's lemma, our results will apply to a wider family of distributions over the random set $R$ than just the law of $R\subseteq_r[n]$; they will hold when $R$ is drawn from any \emph{$r$-spread} distribution.
\begin{definition}
A distribution $\mathcal{R}$ over subsets of $[n]$ is $r$-spread if, for every $S \subseteq [n]$, we have $\Pr_{R \sim \mathcal{R}}[ R \supseteq S] \leq r^{|S|}$. 
\end{definition}
We will only need our lemma in the particular case where $R\subseteq_r[n]$ but we state it in this general form since we believe it might be of some wider interest. In particular, we emphasize that the uniform marginals lemma \cite{GPW2020,Wu2018} mentioned above applies in this same setting (where the random set $R$ is only guaranteed to be $r$-spread rather than having a particular distribution). In the literature on lifting, spreadness is instead referred to as ``blockwise density.'' Spread sets arise also in the context of the sunflower lemma \cite{ALWZ2021} and the fractional expectation threshold conjecture \cite{Talagrand2010,FKNP2021} (see \cite{LMMPZ2022} for work connecting the sunflower lemma to lifting).
\begin{lemma}[Uniformity Lemma]\label{lem:uniformity}
Let $k \geq 1, t \geq 2, r \in [0,1]$ be given. Say that $\bm{X}$ is a random variable in $\B^n$, $D_{\infty}(\bm{X}) \leq k$, and $\mathcal{R}$ is an $r$-spread distribution on subsets of $[n]$. Let $f(x)=2^n\Pr[\bm{X}=x]$ be the renormalized probability mass function of $\bm{X}$. For $R \subseteq [n]$, let $f_R: \B^R \to \mathbb{R}$ be the renormalized probability mass function of $\bm{X}_R$ given by:
\[
f_R(z) = 2^{|R|} \Pr[ \bm{X}_R = z] = \mathbb{E}_{x \sim \B^n}[f(x) \mid x_R=z]
\]
Then for any $\eta \in (0,1)$, 
\[
\Pr_{R \sim \mathcal{R}} \Bigl[ \norm{f_R  - 1 }_t > \sqrt{\frac{Crk(t-1)}{\eta}}\Bigr]  \leq \eta
\]
provided $rk(t-1) < \epsilon$, where $C\in \mathbb{N}, \epsilon > 0$ are universal constants.
\end{lemma}
Note that the constant function $1$ on $\B^R$ is the renormalized probability mass function of the uniform distribution on $\B^R$, so this lemma is giving a bound on the deviation of $\bm{X}_R$ from the uniform distribution. To prove this lemma we need two standard consequences of the hypercontractivity theorem of Bonami \cite{Bonami1970} (see \cite{ODonnell2014}).
\begin{lemma}[$L_2 \to L_t$ Inequality]
For any $f: \B^m \to \mathbb{R}$, $t \geq 2$,
\[
\norm{f}_t^2 \leq \sum_{S \subseteq [m]} (t-1)^{|S|} \widehat{f}(S)^2
\]
\end{lemma}
\begin{lemma}[Level-$\ell$ Inequality]
Let $k \geq 1$, and say that $\bm{X}$ is a random variable in $\B^n$ with $D_\infty(\bm{X}) \leq k$. For $S \subseteq [n]$, let $b_{\bm{X}}(S) := \mathbb{E}[ \chi_S(\bm{X})]$ be the bias of the $S^{th}$ Fourier character on $\bm{X}$. Then for any $\ell \in \{ 1,\ldots, n\}$,
\[
\sum_{|S| = \ell} b_{\bm{X}}(S)^2 \leq (C k )^{\ell}
\]
where $C$ is a universal constant.
\end{lemma}

\begin{proof}[Proof of the Uniformity Lemma]
We will show
\[
\mathbb{E}_{R  \sim\mathcal{R}} \norm{f_R -1}_t^2 \leq 2Crk(t-1)
\]
where $C$ is the constant from the Level-$\ell$ inequality which will yield the lemma via Markov. Observe that for any $S \subseteq R$, $\widehat{f_R - 1}(S) = b_{\bm{X}}(S)$ when $S \neq \emptyset$, and $\widehat{f_R - 1}(\emptyset) = 0$. Applying first the $L_2 \to L_t$ inequality, then the $r$-spreadness condition, then the Level-$\ell$ inequality (with basic rearrangements in between) we have:
\begin{gather*}
\mathbb{E}_{R \sim\mathcal{R}} \norm{f_R -1}_t^2 \leq \mathbb{E}_{R} \sum_{S \subseteq R} (t- 1)^{|S|} \widehat{f_R - 1}(S)^2 = \mathbb{E}_{R} \sum_{ \emptyset \neq S \subseteq R} (t- 1)^{|S|} b_{\bm{X}}(S)^2 \\
= \sum_{ \emptyset \neq S \subseteq [n]} \Pr_R[R \supseteq S] \cdot (t- 1)^{|S|} b_{\bm{X}}(S)^2 \leq \sum_{ \emptyset \neq S \subseteq [n]} r^{|S|}(t-1)^{|S|} b_{\bm{X}}(S)^2\\
= \sum_{\ell=1}^n (r (t-1))^\ell \sum_{S \subseteq [n], |S|=\ell} b_{\bm{X}}(S)^2 
\leq \sum_{\ell = 1}^\infty (Crk(t-1))^\ell \leq 2C rk(t-1)
\end{gather*}
where in the final inequality we require $Crk(t-1) \leq \frac{1}{2}$, which follows from our assumption that $rk(t-1)$ is sufficiently small.
\end{proof}

We can now prove Lemma~\ref{lem:light-patterns}. We will state here the strengthening to $r$-spread sets for the sake of completeness, although we only need the case $R\subseteq_r[n]$ in our main results.
\begin{lemma*}[Light Patterns (Lemma~\ref{lem:light-patterns}), Generalized to Spread Sets]
Let $ k \geq 1, r \in [0,1]$, let $\mathcal{R}$ be an $r$-spread distribution on $2^{[n]}$, and let $\bm{X}$ be a random variable in $\B^n$ with $D_\infty(\bm{X}) \leq k$. There is a universal constant $C$ such that the following holds whenever $r \leq (Ck)^{-2}$: with probability $\geq \frac{63}{64}$ over $R \sim \mathcal{R}$ we have:
\[
|\{ z \in \B^R \mid \Pr[\bm{X}_R = z] \leq 2^{- |R| - 1}\}| \leq 2^{|R|- k}
\]
\end{lemma*}
\begin{proof}
For a fixed $R$, let $Z_R = \{ z \in \B^R \mid \Pr[\bm{X}_R = z] \leq 2^{- |R| - 1}\}$. Set $t = \max\{ k, 2\}$. Then $|f_R(z) - 1|^t \geq 2^{-t}$ for all $z \in Z_R$. Say that $|Z_R| \geq 2^{|R| - k}$. Then we have
\[
\norm{f_R - 1}_t \geq  (2^{-t} \Pr_{z \sim \B^R} [ z \in Z_R])^{1/t} \geq 2^{ -\frac{k}{t} - 1} \geq \frac{1}{4}
\]
Now, applying the uniformity lemma, we have
$\Pr_{\mathcal{R}} [ |Z_R| \geq 2^{|R| - k}] \leq \frac{1}{64}$
provided that $\sqrt{Crk(t-1)} < \frac{1}{4}$ for some universal constant $C$. This holds provided $r \leq (C'k)^{-2}$ for some universal constant $C'$.
\end{proof}

\section{A Nearly Tight Lower Bound}\label{sec:nearly-tight}
Finally in this section we prove:
\begin{theorem}
\label{thm:nearly-tight-parity-game}
For each $d\geq 2$ there is $\epsilon_d>0$ such that
$\mathrm{KW}(\mathrm{Parity}_n)$ does not have a $d$-round communication protocol of cost $m$ when $m\leq \epsilon_d n^{1/(d-1)}$, provided $n$ is sufficiently large. Consequently any depth $d$ circuit computing $\mathrm{Parity}_n$ must have $>2^{\epsilon_d n^{1/(d-1)}}$ wires.
\end{theorem}
Theorem~\ref{thm:nearly-tight-parity-game} will be a consequence of the following improvement to Lemma~\ref{lem:light-patterns}:
\begin{lemma}[Improved Light Patterns]\label{lem:improved-light-patterns}
Let $k \geq 1$, $0 \leq r \leq (512k)^{-1}$, and let $\bm{X}$ be a random variable in $\B^n$, $D_\infty(\bm{X}) \leq k$. With probability $\geq \frac{63}{64}$ over $R \subseteq_r [n]$, we have
\[
| \{ z \in \B^R \mid \Pr[\bm{X}_R = z] \leq 2^{-|R|-2k-2} \}| \leq 2^{|R|-k}
\]
\end{lemma}
Formally this is incomparable to Lemma~\ref{lem:light-patterns}; the important sense in which it is stronger is the weakened requirement $r \leq (Ck)^{-1}$ for a universal constant $C$, whereas Lemma~\ref{lem:light-patterns} required $r \leq (Ck)^{-2}$. We do however have to pay for this by restricting quantitatively the definition of a ``light pattern'' (the kinds of patterns which we aim to bound the number of): in Lemma~\ref{lem:light-patterns} these were patterns whose probabilities decayed by a factor $\leq \frac{1}{2}$ compared to uniform, whereas here they are patterns whose probabilities decayed by $\leq 2^{-2k-2}$. Fortunately, using this more restrictive definition of a light pattern comes at no cost to us in our application to Theorem~\ref{thm:nearly-tight-parity-game}.
Theorem~\ref{thm:nearly-tight-parity-game} is derived from Lemma~\ref{lem:improved-light-patterns} in exactly the same way that Theorem~\ref{thm:strong-parity-game} is derived from Lemma~\ref{lem:light-patterns} (except we now achieve $q= O(kp)$ in our mirror set lemma) and so we will not repeat the argument here. The remainder of this section will be dedicated to proving Lemma~\ref{lem:improved-light-patterns}. 

Let $f_R: \B^R \to \mathbb{R}$ be the renormalized distribution function of $\bm{X}_R$ as in Section~\ref{sec:new-lb}, given by $f_R(z) = 2^{|R|} \Pr[\bm{X}_R = z]$. Our approach will be to show that with probability at least $63/64$ over $R \subseteq_r [n]$,
\begin{gather}\label{eq:harm-bound}
\mathbb{E}_{z \sim \B^R} f_R(z)^{-1} \leq 4 \cdot 2^k
\end{gather}
with the convention $0^{-1} = +\infty$. This immediately yields the bound in Lemma~\ref{lem:improved-light-patterns} by Markov. We will derive this from some general analytic inequalities involving a transform on functions over the hypercube which we will call the ``harmonic mean transform:''
\begin{definition}[Harmonic Mean Transform]
For a function $f: \B^n \to \mathbb{R}_{\geq 0}$ and any $R \subseteq [n]$, define $f_R: \B^R \to \mathbb{R}$ by
\[
f_R(z) = \mathbb{E}_{x\sim \B^n} [f(x) \mid x_R = z]
\]
We then define $\harm{f}: 2^{[n]} \to \mathbb{R}_{\geq 0}$, the ``harmonic mean transform of $f$,''  by
\[
\harm{f}(R) = \Bigl(\mathbb{E}_{z \sim \B^R} f_R(z)^{-1} \Bigr)^{-1}
\]
where we use the conventions $0^{-1} = +\infty$, $(+\infty)^{-1} = 0$.
\end{definition}
Our desired bound in (\ref{eq:harm-bound}) can be immediately recast in terms of the harmonic mean transform of the renormalized distribution function of $\bm{X}$: for $f(x)=2^n\Pr[\bm{X}=x]$, bound~\eqref{eq:harm-bound} becomes $\harm{f}(R)\geq 2^{-k-2}$. In some very rough sense this transform will play a role similar to that of the Fourier transform in the proof of the uniformity lemma in Section~\ref{sec:new-lb}. To be precise, we believe the most accurate analogy is between $\harm{f}(R)$ and the quantity 
$\sum_{\emptyset \neq S \subseteq R} (t-1)^{|S|} \hat{f}(S)^2$ from the uniformity lemma's proof, rather than between $\harm{f}(R)$ and $\hat{f}(R)$. We start with presenting an alternate variational formula for the harmonic mean transform:
\begin{lemma}[Variational Formula]
For all $f: \B^n \to \mathbb{R}_{\geq 0}$, $R \subseteq [n]$,
\[
\harm{f}(R) = \inf_{\substack{w: \B^R \to \mathbb{R}_{\geq 0} \\ \mathbb{E}_{z} w(z) = 1}} \mathbb{E}_{z}\bigl[f_R(z) w(z)^2\bigr]
\]
where each expectation over $z$ is uniform on $\B^R$.
\end{lemma}
\begin{proof}
We may assume $f_R(z) > 0$ for all $z \in \B^R$, otherwise both sides of the identity must be zero. It is easy to see that the right hand side is at most the left by taking $w(z) = \harm{f}(R) f_R(z)^{-1}$. For the interesting direction, take any admissible $w$ and apply Cauchy-Schwarz:
\[
1 = (\mathbb{E}_z w(z))^2 = (\mathbb{E}_z (\sqrt{f_R(z)}w(z)) f_R(z)^{-1/2})^2 \leq \frac{\mathbb{E}_z[f_R(z)w(z)^2]}{\harm{f}(R)}
\]
\end{proof}
Three useful corollaries follow immediately:
\begin{corollary}\label{cor:harm-properties}
For all $f,g,h: \B^n \to \mathbb{R}_{\geq 0}$ we have:
\begin{enumerate}
  \item The map $f \mapsto \harm{f}$ is concave: if $f = \lambda g + (1-\lambda)h$ for some $\lambda \in [0,1]$, then $\harm{f} \geq \lambda \harm{g} + (1-\lambda)\harm{h}$ pointwise.
  \item $\harm{f}(R) \geq \harm{f}(S)$ whenever $R \subseteq S$.
  \item $\harm{f}(R)\leq \mathbb{E} f = \harm{f}(\emptyset)$.
\end{enumerate}
\end{corollary}
\begin{proof}
The first follows since the variational formula expresses $f \mapsto \harm{f}(R)$ as an infimum over linear functions of $f$. The second follows since any admissible $w: \B^R \to \mathbb{R}_{\geq 0}$ in the variational form for $\harm{f}(R)$ extends to an admissible $w': \B^S \to \mathbb{R}_{\geq 0}$ in the variational form for $\harm{f}(S)$ achieving the same value: take $w'(z) = w(z_R)$. The third is immediate from the second.
\end{proof}

We now arrive at the central analytic inequality we use for the harmonic mean transform:
\begin{lemma}\label{lem:analytic}
For any $f: \B^n \to \mathbb{R}_{\geq 0}$,
\[
\mathbb{E}_{R \subseteq_{1/4} [n]} \sqrt{\harm{f}(R)} \geq \mathbb{E}_{x \sim \B^n} \sqrt{f(x)}
\]

\end{lemma}
\begin{proof}
We prove by induction on $n$. For $n=1$, let $a = f(0), b = f(1)$. In this case
\[
\harm{f}(\emptyset) = \mathbb{E}f = \frac{a + b}{2}, \quad \harm{f}(\{1\}) = \frac{2ab}{a+b}
\]
so we aim to prove:
\[
\frac{3}{4} \sqrt{\frac{a+b}{2}} + \frac{1}{4} \sqrt{\frac{2ab}{a+b}} \geq \frac{\sqrt{a} + \sqrt{b}}{2}
\]
with the convention here and below that $\frac{2ab}{a+b}=0$ when $a = b=0$. The case $a=b=0$ is immediate, so we assume $a+b > 0$. To verify the inequality, set $u = \sqrt{a},v = \sqrt{b}$, $c = \sqrt{(u^2 + v^2)/2}$, $d = (u+v)/2$, and rewrite it as:
\[
\frac{3}{4} c  + \frac{uv}{4c} \geq d
\]  
Since $uv = 2d^2 - c^2$, $c > 0$, we have:
\[
\frac{3}{4}c + \frac{uv}{4c} - d = \frac{3}{4}c + \frac{2d^2-c^2}{4c} - d = \frac{(c-d)^2}{2c} \geq 0
\]
For the inductive step, let $n\geq2$ and define
\[
f_0(x)=f(x,0),\qquad f_1(x)=f(x,1),
\qquad f_{\mathrm{avg}}=\frac{f_0+f_1}{2}.
\]
Observe that $\harm{f}_{\mathrm{avg}} \geq \frac{1}{2}(\harm{f}_0 + \harm{f}_1)$ pointwise by concavity. For every $S\subseteq[n-1]$ we have:
\[
\harm{f}(S)=\harm{f}_{\mathrm{avg}}(S),
\qquad
\harm{f}(S\cup\{n\})
=\frac{2\harm{f}_0(S)\harm{f}_1(S)}
{\harm{f}_0(S)+\harm{f}_1(S)},
\]
where the fraction on the right is interpreted as zero when its denominator is zero. Sampling $S\subseteq_{1/4}[n-1]$ and conditioning on whether $n\in R$, we obtain
\[
\begin{aligned}
\mathbb E_{R\subseteq_{1/4}[n]}\sqrt{\harm{f}(R)}
&=\mathbb E_S\left[
\frac34\sqrt{\harm{f}_{\mathrm{avg}}(S)}
+\frac14\sqrt{
\frac{2\harm{f}_0(S)\harm{f}_1(S)}
{\harm{f}_0(S)+\harm{f}_1(S)}
}
\right]\\
&\geq\mathbb E_S\left[
\frac34\sqrt{\frac{\harm{f}_0(S)+\harm{f}_1(S)}2}
+\frac14\sqrt{
\frac{2\harm{f}_0(S)\harm{f}_1(S)}
{\harm{f}_0(S)+\harm{f}_1(S)}
}
\right]\\
&\geq\frac12\mathbb E_S\sqrt{\harm{f}_0(S)}
+\frac12\mathbb E_S\sqrt{\harm{f}_1(S)}\\
&\geq\frac12\mathbb E_{x\sim\B^{n-1}}\sqrt{f_0(x)}
+\frac12\mathbb E_{x\sim\B^{n-1}}\sqrt{f_1(x)}\\
&=\mathbb E_{x\sim\B^n}\sqrt{f(x)}.
\end{aligned}
\]
where the first inequality uses concavity, the second applies the $n=1$ case, and the third applies the inductive hypothesis.
\end{proof}

To apply this inequality to study a sparser random set $R \subseteq_r [n]$, we apply an altered form of the same standard transfer principle for downward closed families used in the proof of the shattering lemma \cite[Lemma~4.3.7]{Zhao2024}: for any nonempty downward-closed family of sets $\mathcal{A} \subseteq 2^{[n]}$ and $0\leq r\leq 1/20$,
\[
\Pr_{R \subseteq_r [n]}[ R \in \mathcal{A}] \geq \left(\Pr_{S\subseteq_{1/4}[n]}[S\in\mathcal{A}]\right)^{5r}.
\]
The proof is exactly the same as the case we presented during the proof of the shattering lemma (in that case we used $1/2$ in place of $1/4$). At this point we are ready to prove Lemma~\ref{lem:improved-light-patterns}:

\begin{proof}[Proof of Lemma~\ref{lem:improved-light-patterns}]
As per the preceding discussion, it suffices to establish that, with probability $\geq \frac{63}{64}$ over $R \subseteq_r [n]$, we have $\harm{f}(R) \geq 2^{-k-2}$, where $f$ is the renormalized distribution function of $\bm{X}$. Let $\mathcal{A} \subseteq 2^{[n]}$ be the family of sets $R$ such that $\harm{f}(R) \geq 2^{-k-2}$. By the second point in Corollary~\ref{cor:harm-properties}, we know that $\mathcal{A}$ is a downward-closed family. We will prove that $\Pr_{R \subseteq_{1/4}}[ R \in \mathcal{A}] \geq 2^{-k/2-1}$. Applying the above transfer principle and the standard estimate $1- 2^{-t} \leq t$ valid for all $t \geq 0$ we have
\[
\Pr_{R \subseteq_r [n]}[ R \notin \mathcal{A}] \leq 1- (2^{-k/2-1})^{5r} \leq 5r(k/2+1) \leq \frac{1}{64}
\]
where the final inequality simply applies our assumptions $r \leq (512 k)^{-1}, k \geq 1$. It remains only to prove  $\Pr_{R \subseteq_{1/4}}[ R \in \mathcal{A}] \geq 2^{-k/2-1}$. Applying Lemma~\ref{lem:analytic}:
\[
\mathbb{E}_{R \subseteq_{1/4} [n]} \sqrt{\harm{f}(R)} \geq \mathbb{E}_{x \sim \B^n} \sqrt{f(x)} \geq 2^{- D_\infty(\bm{X})/2} \mathbb{E}_x f(x) \geq 2^{-k/2}
\]
Hence, by Markov,
\[
\Pr_{R \subseteq_{1/4}[n]} [\sqrt{\harm{f}(R)} < \frac{1}{2} 2^{-k/2}] \leq 1- 2^{-k/2 - 1} 
\]
using the fact (point 3 of Corollary~\ref{cor:harm-properties}) that $\sqrt{\harm{f}(R)} \leq 1 = \mathbb{E}f$ for all $R$.
\end{proof}

\section*{Acknowledgments}
The author thanks Mika G\"o\"os, Toniann Pitassi, Artur Riazanov, and Avi Wigderson for their comments on an initial draft of this manuscript.

\renewcommand{\refname}{References }


\begin{thebibliography}{dRMNPRV20}
\bibitem[Paj85]{Pajor1985}
Alain Pajor.
\emph{Sous-espaces $\ell_1^n$ des espaces de Banach}.
Travaux en cours, Hermann, Paris, 1985.

\bibitem[Sau72]{Sauer1972}
Norbert Sauer.
\emph{On the Density of Families of Sets}.
Journal of Combinatorial Theory, Series A, 13(1):145--147, 1972.
\url{https://doi.org/10.1016/0097-3165(72)90019-2}.

\bibitem[She72]{Shelah1972}
Saharon Shelah.
\emph{A Combinatorial Problem; Stability and Order for Models and Theories in Infinitary Languages}.
Pacific Journal of Mathematics, 41(1):247--261, 1972.
\url{https://doi.org/10.2140/pjm.1972.41.247}.

\bibitem[Zha24]{Zhao2024}
Yufei Zhao.
\emph{Probabilistic Methods in Combinatorics}.
MIT 18.226 lecture notes, Fall 2022. Updated June 18, 2024.
\url{https://yufeizhao.com/pm/probmethod_notes.pdf}.

\bibitem[GRSS24]{GRSS2024}
Mika G\"o\"os, Artur Riazanov, Anastasia Sofronova, and Dmitry Sokolov.
\emph{Top-Down Lower Bounds for Depth-Four Circuits}.
arXiv:2304.02555v2, 2024.
\url{https://arxiv.org/abs/2304.02555}.

\bibitem[KW90]{KW1990}
Mauricio Karchmer and Avi Wigderson.
\emph{Monotone Circuits for Connectivity Require Super-Logarithmic Depth}.
SIAM Journal on Discrete Mathematics, 3(2):255--265, 1990.
\url{https://doi.org/10.1137/0403021}.

\bibitem[KRW95]{KRW1995}
Mauricio Karchmer, Ran Raz, and Avi Wigderson.
\emph{Super-Logarithmic Depth Lower Bounds via the Direct Sum in Communication Complexity}.
Computational Complexity, 5(3--4):191--204, 1995.
\url{https://doi.org/10.1007/BF01206317}.

\bibitem[HJP95]{HJP1995}
Johan H{\aa}stad, Stasys Jukna, and Pavel Pudl{\'a}k.
\emph{Top-Down Lower Bounds for Depth-Three Circuits}.
Computational Complexity, 5(2):99--112, 1995.
\url{https://doi.org/10.1007/BF01268140}.

\bibitem[PPZ99]{PPZ1999}
Ramamohan Paturi, Pavel Pudl{\'a}k, and Francis Zane.
\emph{Satisfiability Coding Lemma}.
Chicago Journal of Theoretical Computer Science, 1999, Article 11, pages 1--19.
\url{https://doi.org/10.4086/cjtcs.1999.011}.

\bibitem[MW19]{MW2019}
Or Meir and Avi Wigderson.
\emph{Prediction from Partial Information and Hindsight, with Application to Circuit Lower Bounds}.
Computational Complexity, 28(2):145--183, 2019.
\url{https://doi.org/10.1007/s00037-019-00177-4}.

\bibitem[FSS84]{FSS1984}
Merrick Furst, James B. Saxe, and Michael Sipser.
\emph{Parity, Circuits, and the Polynomial-Time Hierarchy}.
Mathematical Systems Theory, 17:13--27, 1984.
\url{https://doi.org/10.1007/BF01744431}.

\bibitem[Ajt83]{Ajtai1983}
Mikl{\'o}s Ajtai.
\emph{$\Sigma^1_1$-Formulae on Finite Structures}.
Annals of Pure and Applied Logic, 24(1):1--48, 1983.
\url{https://doi.org/10.1016/0168-0072(83)90038-6}.

\bibitem[Yao85]{Yao1985}
Andrew Chi-Chih Yao.
\emph{Separating the Polynomial-Time Hierarchy by Oracles}.
In Proceedings of the 26th Annual Symposium on Foundations of Computer Science (FOCS), pages 1--10, 1985.
\url{https://doi.org/10.1109/SFCS.1985.49}.

\bibitem[Has86]{Hastad1986}
Johan H{\aa}stad.
\emph{Almost Optimal Lower Bounds for Small Depth Circuits}.
In Proceedings of the 18th Annual ACM Symposium on Theory of Computing (STOC), pages 6--20, 1986.
\url{https://doi.org/10.1145/12130.12132}.

\bibitem[Raz87]{Razborov1987}
Alexander A. Razborov.
\emph{Lower Bounds on the Size of Bounded Depth Circuits over a Complete Basis with Logical Addition}.
Mathematical Notes of the Academy of Sciences of the USSR, 41(4):333--338, 1987.
\url{https://doi.org/10.1007/BF01137685}.

\bibitem[Smo87]{Smolensky1987}
Roman Smolensky.
\emph{Algebraic Methods in the Theory of Lower Bounds for Boolean Circuit Complexity}.
In Proceedings of the 19th Annual ACM Symposium on Theory of Computing (STOC), pages 77--82, 1987.
\url{https://doi.org/10.1145/28395.28404}.

\bibitem[Bon70]{Bonami1970}
Aline Bonami.
\emph{{\'E}tude des coefficients de Fourier des fonctions de $L^p(G)$}.
Annales de l'Institut Fourier, 20(2):335--402, 1970.
\url{https://doi.org/10.5802/aif.357}.

\bibitem[O'D14]{ODonnell2014}
Ryan O'Donnell.
\emph{Analysis of Boolean Functions}.
Cambridge University Press, 2014.
\url{https://doi.org/10.1017/CBO9781139814782}.

\bibitem[FKNP21]{FKNP2021}
Keith Frankston, Jeff Kahn, Bhargav Narayanan, and Jinyoung Park.
\emph{Thresholds versus Fractional Expectation-Thresholds}.
Annals of Mathematics, 194(2):475--495, 2021.
\url{https://doi.org/10.4007/annals.2021.194.2.2}.

\bibitem[Wu18]{Wu2018}
Xinyu Wu.
\emph{The uniform marginals lemma in [GPW17]}.
Expository note, 2018.
\url{https://www.contrib.andrew.cmu.edu/~xinyuw1/papers/uniform-marginals-lemma.pdf}.

\bibitem[GPW20]{GPW2020}
Mika G\"o\"os, Toniann Pitassi, and Thomas Watson.
\emph{Query-to-Communication Lifting for BPP}.
SIAM Journal on Computing, 2020. Preliminary version in FOCS 2017.
\url{https://courses.cs.washington.edu/courses/cse599i/24sp/papers/GoosPitassiWatson_2020_journal.pdf}.

\bibitem[Tal10]{Talagrand2010}
Michel Talagrand.
\emph{Are Many Small Sets Explicitly Small?}
In Proceedings of the 42nd ACM Symposium on Theory of Computing (STOC), pages 13--36, 2010.
\url{https://doi.org/10.1145/1806689.1806693}.

\bibitem[ALWZ21]{ALWZ2021}
Ryan Alweiss, Shachar Lovett, Kewen Wu, and Jiapeng Zhang.
\emph{Improved Bounds for the Sunflower Lemma}.
Annals of Mathematics, 194(3):795--815, 2021.
\url{https://doi.org/10.4007/annals.2021.194.3.5}.

\bibitem[LMMPZ22]{LMMPZ2022}
Shachar Lovett, Raghu Meka, Ian Mertz, Toniann Pitassi, and Jiapeng Zhang.
\emph{Lifting with Sunflowers}.
In 13th Innovations in Theoretical Computer Science Conference (ITCS), LIPIcs 215, pages 104:1--104:24, 2022.
\url{https://doi.org/10.4230/LIPIcs.ITCS.2022.104}.

\bibitem[KN97]{KN1997}
Eyal Kushilevitz and Noam Nisan.
\emph{Communication Complexity}.
Cambridge University Press, 1997.
\url{https://doi.org/10.1017/CBO9780511574948}.

\bibitem[RM99]{RM1999}
Ran Raz and Pierre McKenzie.
\emph{Separation of the Monotone NC Hierarchy}.
Combinatorica, 19(3):403--435, 1999. Preliminary version in FOCS 1997.
\url{https://doi.org/10.1007/s004930050062}.

\bibitem[GPW18]{GPW2018}
Mika G\"o\"os, Toniann Pitassi, and Thomas Watson.
\emph{Deterministic Communication vs. Partition Number}.
SIAM Journal on Computing, 47(6):2435--2450, 2018. Preliminary version in FOCS 2015.
\url{https://doi.org/10.1137/16M1059369}.

\bibitem[GP18]{GP2018}
Mika G\"o\"os and Toniann Pitassi.
\emph{Communication Lower Bounds via Critical Block Sensitivity}.
SIAM Journal on Computing, 47(5):1778--1806, 2018. Preliminary version in STOC 2014.
\url{https://doi.org/10.1137/16M1082007}.

\bibitem[dRNV16]{dRNV2016}
Susanna F. de Rezende, Jakob Nordstr\"om, and Marc Vinyals.
\emph{How Limited Interaction Hinders Real Communication (and What It Means for Proof and Circuit Complexity)}.
In Proceedings of the 57th Annual IEEE Symposium on Foundations of Computer Science (FOCS), pages 295--304, 2016.
\url{https://doi.org/10.1109/FOCS.2016.40}.

\bibitem[PR17]{PR2017}
Toniann Pitassi and Robert Robere.
\emph{Strongly Exponential Lower Bounds for Monotone Computation}.
In Proceedings of the 49th Annual ACM SIGACT Symposium on Theory of Computing (STOC), pages 1246--1255, 2017.
\url{https://doi.org/10.1145/3055399.3055478}.

\bibitem[GGKS20]{GGKS2020}
Ankit Garg, Mika G\"o\"os, Pritish Kamath, and Dmitry Sokolov.
\emph{Monotone Circuit Lower Bounds from Resolution}.
Theory of Computing, 16(13):1--30, 2020. Preliminary version in STOC 2018.
\url{https://doi.org/10.4086/toc.2020.v016a013}.

\bibitem[GKRS19]{GKRS2019}
Mika G\"o\"os, Pritish Kamath, Robert Robere, and Dmitry Sokolov.
\emph{Adventures in Monotone Complexity and TFNP}.
In 10th Innovations in Theoretical Computer Science Conference (ITCS), LIPIcs 124, pages 38:1--38:19, 2019.
\url{https://doi.org/10.4230/LIPIcs.ITCS.2019.38}.

\bibitem[dRMNPRV20]{dRMNPRV2020}
Susanna F. de Rezende, Or Meir, Jakob Nordstr\"om, Toniann Pitassi, Robert Robere, and Marc Vinyals.
\emph{Lifting with Simple Gadgets and Applications to Circuit and Proof Complexity}.
In Proceedings of the 61st Annual IEEE Symposium on Foundations of Computer Science (FOCS), pages 24--30, 2020.
\url{https://doi.org/10.1109/FOCS46700.2020.00011}.

\bibitem[dRV25]{dRV2025}
Susanna F. de Rezende and Marc Vinyals.
\emph{Lifting with Colourful Sunflowers}.
In 40th Computational Complexity Conference (CCC), LIPIcs 339, pages 36:1--36:19, 2025.
\url{https://doi.org/10.4230/LIPIcs.CCC.2025.36}.

\bibitem[KR26]{KR2026}
G\"ulce Karde\c{s} and Benjamin Rossman.
\emph{On Top-Down and Local Lower Bounds for $\mathrm{AC}^0$ Circuits}.
arXiv:2609.01759, 2026.
\url{https://arxiv.org/abs/2609.01759}.
\bibitem[Sip84]{Sipser1984}
Michael Sipser.
\emph{A Topological View of Some Problems in Complexity Theory}.
In Mathematical Foundations of Computer Science (MFCS), Lecture Notes in Computer Science 176, pages 567--572. Springer, 1984.
\url{https://doi.org/10.1007/BFb0030341}.

\end{thebibliography}
\end{document}